%% file: main.tex
\documentclass{article}

\usepackage{neurips}

\usepackage[utf8]{inputenc}
\usepackage[T1]{fontenc}
\usepackage{hyperref}
\usepackage{url}
\usepackage{amsmath, amssymb, amsthm}
\usepackage{enumitem}
\usepackage{graphicx}
\usepackage{booktabs}

\DeclareMathOperator*{\argmax}{arg\,max}
\newcommand{\VoI}{\operatorname{VoI}}
\newcommand{\bpost}{\hat{b}}
\newcommand{\E}{\mathbb{E}}
\newcommand{\R}{\mathbb{R}}
\newcommand{\Dg}{D_g}
\newcommand{\Dh}{D_h}

\newtheorem{theorem}{Theorem}
\newtheorem{definition}[theorem]{Definition}
\newtheorem{proposition}[theorem]{Proposition}
\newtheorem{corollary}[theorem]{Corollary}
\newtheorem{remark}[theorem]{Remark}
\newtheorem{theoremA}{Theorem}

\title{All Substitution Is Local}

\author{
    Nidhish Shah \thanks{Equal contribution.} \\
    Independent Researcher \\
    \texttt{shah@nidhish.dev} \\
\And
    Shaurjya Mandal \footnotemark[1] \\
    Harvard University \\
    \texttt{smandal2@mgh.harvard.edu} \\
\And
    Asfandyar Azhar \\
    University of Oxford \\
    \texttt{asfandyar.azhar@st-annes.ox.ac.uk} \\
}

\begin{document}

\maketitle

\begin{abstract}
When does consulting one information source raise the value of another, and when does it diminish it? We study this question for Bayesian decision-makers facing finite actions. The interaction decomposes into two opposing forces: a complement force, measuring how one source moves beliefs to where the other becomes more useful, and a substitute force, measuring how much the current decision is resolved. Their balance obeys a localization principle: substitution requires an observation to cross a decision boundary, though crossing alone does not guarantee it. Whenever posteriors remain inside the current decision region, the substitute force vanishes and sources are guaranteed to complement each other, even when one source cannot, on its own, change the decision. The results hold for arbitrarily correlated sources and are formalized in Lean 4.\footnote{Code and proofs: \url{https://github.com/nidhishs/all-substitution-is-local}.} Substitution is confined to the thin boundaries where decisions change. Everywhere else, information cooperates.
\end{abstract}

\input{chapters/introduction}
\input{chapters/preliminaries}
\input{chapters/methodology}
\input{chapters/illustration}
\input{chapters/related}
\input{chapters/conclusion}

\bibliographystyle{plainnat}
\bibliography{references}

\clearpage
\input{chapters/appendix}

\end{document}

%% file: chapters/introduction.tex
\section{Introduction}
\label{sec:intro}

When an agent can consult multiple information sources before making a decision, a natural question is: does consulting one source make another more or less valuable?
Sources that make each other more valuable are \emph{complements}; sources that diminish each other's value are \emph{substitutes}.

\citet{ChenWaggoner2017} formalized this distinction via the pointwise interaction $\Delta\VoI(j \mid i, b)$: the change in channel~$j$'s value of information at belief~$b$ after observing channel~$i$.
Their work established global conditions (across all beliefs) for substitution and complementarity. Our contribution is to localize these conditions: we characterize \emph{which beliefs} give rise to complementarity and which to substitution, for a fixed pair of channels and a fixed decision problem.
Concretely, we establish three results:
\begin{enumerate}
  \item \textbf{Bregman decomposition} (Proposition~\ref{prop:decomposition}): $\Delta\VoI(j \mid i, b)$ decomposes into a non-negative complement force $\E[D_g(\bpost, b)]$ minus a non-negative substitute force $\E[D_h(\bpost, b)]$, where $D_g$ and $D_h$ are Bregman divergences of auxiliary convex functions.
  \item \textbf{Interior complementarity} (Theorem~\ref{thm:interior}): If all of channel~$i$'s posteriors lie within the current decision region, then $\Delta\VoI(j \mid i, b) \geq 0$.
  \item \textbf{Converse} (Theorem~\ref{thm:converse}): If $\Delta\VoI(j \mid i, b) < 0$, at least one posterior of channel~$i$ lies in a different decision region from~$b$.
\end{enumerate}

Boundary crossing is necessary for substitution but not sufficient: channels can cross decision boundaries and still complement each other. Section~\ref{sec:setup} formalizes the setting; Section~\ref{sec:localization} presents the main results; Section~\ref{sec:example} illustrates them on a concrete example.

\begin{figure}[htbp]
\centering
\includegraphics[width=\textwidth]{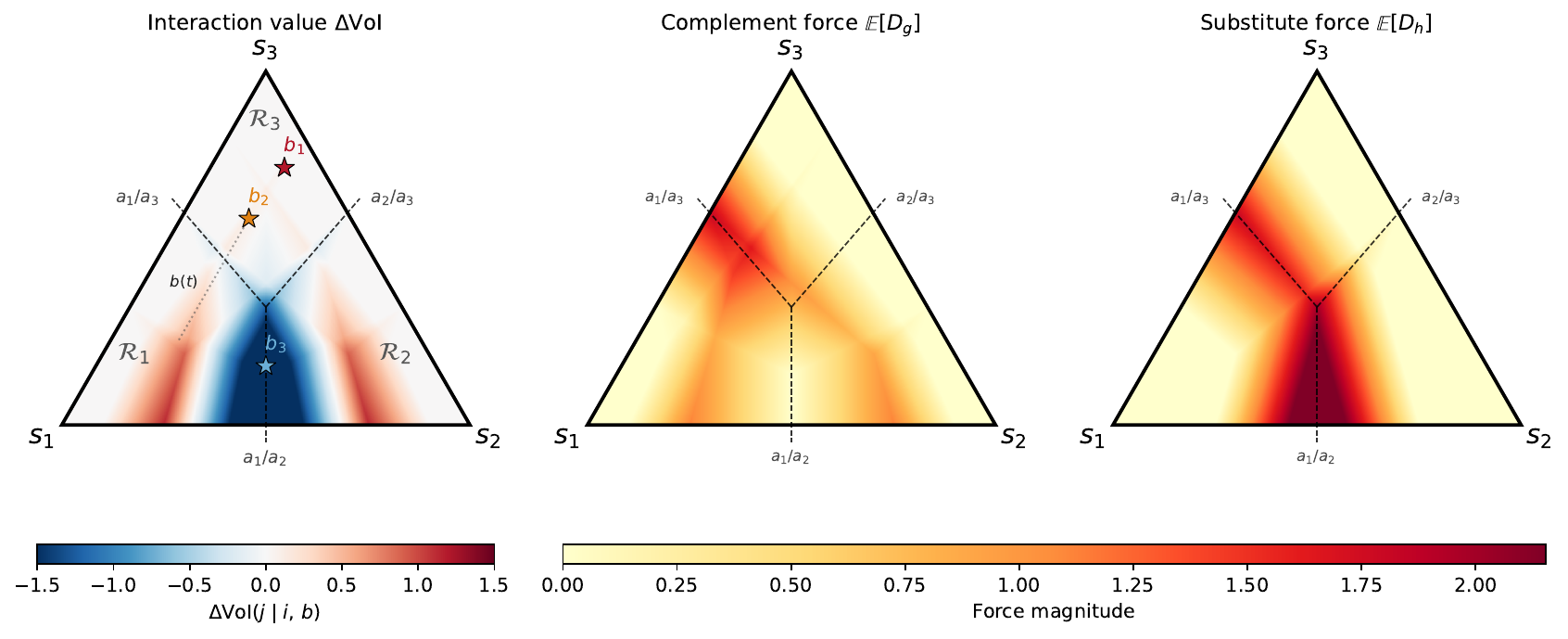}
\caption{Belief-simplex geometry of the interaction.
\textit{Left}: the sign of $\Delta\VoI$ partitions the simplex into complement and substitute regions.
\textit{Middle, Right}: the two constituent forces; the complement force is diffuse across the interior while the substitute force concentrates sharply along decision boundaries. Their difference determines the sign of $\Delta\VoI$.}
\label{fig:hero}
\end{figure}

%% file: chapters/preliminaries.tex
\section{Preliminaries}
\label{sec:setup}

We consider a finite-action Bayesian decision problem and define the second-order interaction between information channels that is the object of study.

\paragraph{Decision problem.}
There are $K$ states $S = \{s_1, \ldots, s_K\}$, a finite action set $A$, and a reward function $R : A \times S \to \R$.
A \emph{belief} $b \in \Delta^{K-1}$ is a probability distribution over states.
The optimal expected reward at belief $b$ is
\[
  V(b) = \max_{a \in A} \sum_s R(a, s)\, b(s),
\]
the upper envelope of $|A|$ linear functions.
$V$ is piecewise-linear and convex (PWLC).
Each \emph{decision region} $\mathcal{R}_a = \{b : a \in \argmax_{a'} \sum_s R(a', s)\, b(s)\}$ is a convex polytope; their interiors partition $\Delta^{K-1}$.
Within region $\mathcal{R}_a$, the value function is linear: $V(b) = r_a \cdot b$ where $r_a = (R(a, s_1), \ldots, R(a, s_K))$.

\paragraph{Channels.}
A \emph{channel} $i$ has outcome set $O_i$ and likelihood kernel $P(\cdot \mid s, i)$.
The \emph{value of information} (VoI) of channel $i$ at belief $b$ is
\[
  \VoI(i \mid b) = \E_{o \sim P(\cdot \mid b, i)}\bigl[V(\bpost(i, o))\bigr] - V(b),
\]
where $\bpost(i, o) \propto P(o \mid s, i)\, b(s)$ is the posterior after outcome $o$.
Channels $i$ and $j$ are \emph{conditionally independent} given the state. This ensures that Definition~\ref{def:delta} captures the actual interaction; the mathematical results hold without this assumption.

\begin{definition}[Second-order interaction]\label{def:delta}
The \emph{interaction} between channels $i$ and $j$ at belief $b$ is
\[
  \Delta\VoI(j \mid i, b)
  \;=\;
  \E_{o_i}\bigl[\VoI(j \mid \bpost(i, o_i))\bigr]
  \;-\;
  \VoI(j \mid b).
\]
When $\Delta\VoI > 0$, observing channel~$i$ \emph{increases} the value of subsequently observing channel~$j$ (complements); when $\Delta\VoI < 0$, it decreases it (substitutes).
\end{definition}

\paragraph{Auxiliary functions.}
Define:
\begin{align}
  g(b) &= \E_{o_j}\bigl[V(\bpost(j, o_j))\bigr], \label{eq:g} \\
  h(b) &= V(b). \label{eq:h}
\end{align}
Both $g$ and $h$ are convex: $g$ is an expectation of the convex function $V$, and $h = V$ is a maximum of linear functions (PWLC).
Note that $\VoI(j \mid b) = g(b) - h(b)$ by definition.
These two functions are the building blocks of the Bregman decomposition in Proposition~\ref{prop:decomposition}.

\paragraph{Bregman divergence.}
The Bregman divergence of a convex function $\phi$ is
$D_\phi(p, q) = \phi(p) - \phi(q) - \langle \xi, p - q\rangle$,
where $\xi \in \partial\phi(q)$ is any subgradient.

\paragraph{Martingale identity.}
For any convex $\phi$,
\begin{equation}
  \label{eq:bregman-jensen}
  \E[\phi(\bpost)] - \phi(b)
  \;=\;
  \E\bigl[D_\phi(\bpost, b)\bigr],
\end{equation}
since $\E[\bpost] = b$ (Bayesian martingale) causes the linear correction term $\langle \nabla\phi(b),\, \E[\bpost] - b \rangle$ to vanish.
At points where $\phi$ is non-differentiable (e.g., kinks of $h = V$ at decision boundaries), the identity holds for any choice of subgradient $\xi \in \partial\phi(b)$, since the correction $\langle \xi, \E[\bpost] - b \rangle = 0$ regardless.

%% file: chapters/methodology.tex
\section{The Localization Principle}
\label{sec:localization}

\paragraph{Two Opposing Forces.}

\begin{proposition}[Bregman decomposition]
\label{prop:decomposition}
\[
  \Delta\VoI(j \mid i, b)
  \;=\;
  \E\bigl[\Dg(\bpost, b)\bigr]
  \;-\;
  \E\bigl[\Dh(\bpost, b)\bigr],
\]
where $\bpost = \bpost(i, o)$, the functions $g$ and $h$ are defined in \eqref{eq:g}--\eqref{eq:h}, and both expected Bregman divergences are non-negative.
\end{proposition}

\begin{proof}
Since $\VoI(j \mid b') = g(b') - h(b')$ by \eqref{eq:g}--\eqref{eq:h}, expanding Definition~\ref{def:delta} gives
\[
  \Delta\VoI(j \mid i, b)
  \;=\;
  \bigl[\E[g(\bpost)] - g(b)\bigr]
  \;-\;
  \bigl[\E[h(\bpost)] - h(b)\bigr].
\]
By the martingale identity~\eqref{eq:bregman-jensen}, each bracketed expression equals the corresponding expected Bregman divergence.
Non-negativity of $\E[\Dg(\bpost, b)]$ follows from convexity of $g$, and non-negativity of $\E[\Dh(\bpost, b)]$ from convexity of $h$, each by Jensen's inequality.
\end{proof}

The decomposition reveals a tug-of-war: the term $\E[\Dh(\bpost, b)]$ is the \emph{substitute force}, measuring how much channel~$i$ resolves the current decision by generating posteriors that straddle a decision boundary, while $\E[\Dg(\bpost, b)]$ is the \emph{complement force}, measuring how much channel~$i$ enhances the future relevance of channel~$j$.
Under conditional independence, $g(\bpost(i, o_i))$ equals the actual expected post-$j$ value $\E_{o_j \mid o_i}[V(\bpost(i,j,o_i,o_j))]$, so Definition~\ref{def:delta} captures the real interaction and the forces have direct economic content. Without conditional independence, the algebra and localization theorems still hold, but the individual forces may not reflect the true conditional distributions.

\begin{corollary}[Independence from channel ordering and costs]
\label{cor:cost}
The interaction $\Delta\VoI(j \mid i, b)$ is determined entirely by the two likelihood kernels and the reward structure; it does not depend on any separate per-use costs that might be attached to channels $i$ or $j$.
\end{corollary}

\begin{proof}
No cost terms appear in Definition~\ref{def:delta}.
Equivalently, adding a constant to $g$ or $h$ does not change the Bregman divergence, since $D_{\phi+c}(p,q) = D_\phi(p,q)$ for any constant~$c$.
\end{proof}

\paragraph{When Posteriors Stay Inside The Decision Region.}

\begin{theorem}[Interior complementarity]
\label{thm:interior}
Suppose all posteriors $\bpost(i, o)$ lie in the closed decision region~$\mathcal{R}_a$.
Then $\Delta\VoI(j \mid i, b) \geq 0$.
This holds for arbitrarily correlated channels $i$ and $j$.
\end{theorem}

\begin{proof}
Since $\mathcal{R}_a$ is convex and $b = \sum_o P(o \mid b)\, \bpost(i, o)$, the prior $b$ also lies in $\mathcal{R}_a$.
Within $\mathcal{R}_a$, the value function satisfies $h(b') = r_a \cdot b'$ for all $b'$ in the region.
Because $h$ is linear on $\mathcal{R}_a$, $\E[h(\bpost)] = h(\E[\bpost]) = h(b)$ by the martingale property, so the Jensen gap of $h$ vanishes.
Meanwhile $g$ is convex, so $\E[\Dg(\bpost, b)] \geq 0$ by Jensen's inequality.
By Proposition~\ref{prop:decomposition},
\[
  \Delta\VoI(j \mid i, b) = \E[\Dg(\bpost, b)] - 0 \geq 0. \qedhere
\]
\end{proof}

\begin{remark}
\label{rem:irrelevant}
The proof above shows that when all posteriors remain in $\mathcal{R}_a$, channel~$i$ is decision-irrelevant: $\VoI(i \mid b) = 0$.
Despite being unable to improve the current decision on its own, channel~$i$ weakly amplifies the value of channel~$j$.
The mechanism is informational, not decisional: channel~$i$ may sharpen the posterior toward regions where $j$'s signals are more discriminating, even when no action change results from~$i$ alone.
(The hypothesis uses the closed region $\mathcal{R}_a$, so the Jensen-gap argument extends to boundary faces.)
\end{remark}

\paragraph{When Posteriors Cross A Decision Boundary.}

\begin{theorem}[Converse: substitution requires boundary-crossing]
\label{thm:converse}
If $\Delta\VoI(j \mid i, b) < 0$, then at least one posterior $\bpost(i, o)$ lies in a different
decision region from~$b$.
This holds for arbitrarily correlated channels $i$ and $j$.
\end{theorem}

\begin{proof}
Since $g$ is convex (it is the expected maximum of linear functions),
Proposition~\ref{prop:decomposition} and $\Delta\VoI < 0$ imply
\[
  \E[\Dh(\bpost, b)] \;>\; \E[\Dg(\bpost, b)] \;\geq\; 0,
\]
so $\E[\Dh(\bpost, b)] > 0$.
By the martingale property~\eqref{eq:bregman-jensen},
$\E[\Dh(\bpost, b)] = \E[h(\bpost)] - h(b)$.
The linear correction term $\langle \xi, \E[\bpost] - b \rangle$ vanishes for any choice of subgradient $\xi \in \partial h(b)$, because $\E[\bpost] = b$.
So $\E[h(\bpost)] > h(b)$.
Since $h$ is PWLC, this strict inequality requires $h$ to be nonlinear on the convex hull
of the posterior support.
But $h$ is linear within each decision region, so the posterior support must span at least
two distinct decision regions.
\end{proof}

These two conditions leave a gap: a channel can cross a decision boundary and still produce complementarity, so long as the complement force $\E[\Dg]$ dominates the substitute force $\E[\Dh]$.
The worked example in the next section illustrates this.

%% file: chapters/illustration.tex
\section{Illustration}
\label{sec:example}

We illustrate all three results with a single example. The same two channels and the same reward structure produce three qualitatively different interaction regimes at three different beliefs: interior complementarity, boundary complementarity, and substitution. The key distinction is not whether decision boundaries are crossed, but \emph{which} boundaries each channel is relevant to. Consider $K = 3$ states and $|A| = 3$ actions with the parameters in Table~\ref{tab:rewards}.

The decision regions partition the 2-simplex: $\mathcal{R}_1$ (where $a_1$ is optimal) lies near the $s_1$ vertex, $\mathcal{R}_2$ near $s_2$, and $\mathcal{R}_3$ (the safe action) occupies a large central region near~$s_3$ (see Figure~\ref{fig:hero}, left, for the geometry).

\begin{table}[b]
\centering
\begin{tabular}{@{}lccc@{}}
\toprule
       & $s_1$ & $s_2$ & $s_3$ \\
\midrule
$a_1$  & 12    & 0     & 3     \\
$a_2$  & 0     & 12    & 3     \\
$a_3$  & 3     & 3     & 9     \\
\bottomrule
\end{tabular}
\qquad
\begin{tabular}{@{}lccc@{}}
\toprule
$P(o \mid s, i)$ & $s_1$ & $s_2$ & $s_3$ \\
\midrule
$o = 0$ & $3/4$ & $1/4$ & $1/4$ \\
$o = 1$ & $1/4$ & $3/4$ & $3/4$ \\
\bottomrule
\end{tabular}
\qquad
\begin{tabular}{@{}lccc@{}}
\toprule
$P(o \mid s, j)$ & $s_1$ & $s_2$ & $s_3$ \\
\midrule
$o = 0$ & $1/4$ & $3/4$ & $1/4$ \\
$o = 1$ & $3/4$ & $1/4$ & $3/4$ \\
\bottomrule
\end{tabular}
\vspace{1em}
\caption{Decision problem parameters. \textit{Left}: reward matrix $R(a, s)$; $a_3$ is a safe action, $a_1$ and $a_2$ specialize to $s_1$ and $s_2$. \textit{Middle, Right}: likelihood kernels for channels $i$ and $j$. Channel~$i$ partially identifies~$s_1$; channel~$j$ partially identifies~$s_2$. The channels are conditionally independent.}
\label{tab:rewards}
\end{table}

The left panel of Figure~\ref{fig:hero} shows $\Delta\VoI(j \mid i, b)$ across the belief simplex. Complements (red) fill region interiors; substitutes (blue) cluster along decision boundaries. The middle and right panels decompose the interaction: the complement force $\E[\Dg]$ spreads broadly across the simplex, while the substitute force $\E[\Dh]$ concentrates sharply along decision boundaries. The substitute force equals the expected regret of the currently optimal action at the posterior beliefs (Appendix~\ref{app:regret}). We walk through three marked beliefs that illustrate each main result.

\paragraph{Interior complementarity.}
Consider $b_1 = (1/11,\; 2/11,\; 8/11)$, in the interior of $\mathcal{R}_3$. Theorem~\ref{thm:interior} applies directly. Channel~$i$ is decision-irrelevant: applying Bayes' rule, both posteriors $(3/13,\, 2/13,\, 8/13)$ and $(1/31,\, 6/31,\, 24/31)$ remain in $\mathcal{R}_3$, so $\VoI(i) = 0$ and $\E[\Dh] = 0$ exactly. The substitute force vanishes, exactly as Theorem~\ref{thm:interior} predicts.

Yet channel~$j$ is useful: its $o{=}0$ posterior crosses into $\mathcal{R}_2$, giving $\VoI(j) = 3/44$. After observing the ``useless'' channel~$i$, $j$'s value rises to $15/176$. The mechanism: channel~$i$'s ``probably not $s_1$'' signal ($o{=}1$, probability $31/44$) concentrates the remaining uncertainty on ``$s_2$ vs $s_3$'', which is exactly what channel~$j$ resolves.

\paragraph{Boundary complementarity.}
Consider $b_2 = (1/4,\; 1/6,\; 7/12)$, in $\mathcal{R}_3$ but closer to the decision boundaries. Both channels cross boundaries, so Theorem~\ref{thm:interior}'s hypothesis is violated. The decomposition (Proposition~\ref{prop:decomposition}) reveals what happens beyond the theorem.

Both channels have positive VoI; each crosses a decision boundary. Yet they complement: $\VoI(j)$ is amplified $2.25\times$ by observing~$i$ first. The substitute force $\E[\Dh] = 11/16$ is strictly positive (the boundary is crossed), but the complement force $\E[\Dg] = 49/64$ dominates. Channel~$i$ crosses the $a_3/a_1$ boundary while channel~$j$ crosses $a_3/a_2$, \emph{different} boundaries, orthogonal uncertainties.

\paragraph{Substitution.}
Consider $b_3 = (5/12,\; 5/12,\; 1/6)$, on the $b_1 = b_2$ boundary. Both channels cross the \emph{same} boundary ($a_1/a_2$). Channel~$i$'s posteriors $(15/22,\, 5/22,\, 2/22)$ and $(5/26,\, 15/26,\, 6/26)$ land in $\mathcal{R}_1$ and $\mathcal{R}_2$ respectively, straddling the same $a_1/a_2$ boundary. After $o{=}0$ (a strong $s_1$ signal), $j$'s VoI collapses to~$0$: channel~$i$ has resolved the very decision that~$j$ was relevant to. Table~\ref{tab:results} collects the numerical values at all three beliefs.

\begin{table}[b]
\centering
\small
\begin{tabular}{@{}lccc@{}}
\toprule
 & $b_1 = (\tfrac{1}{11},\, \tfrac{2}{11},\, \tfrac{8}{11})$
 & $b_2 = (\tfrac{1}{4},\, \tfrac{1}{6},\, \tfrac{7}{12})$
 & $b_3 = (\tfrac{5}{12},\, \tfrac{5}{12},\, \tfrac{1}{6})$ \\
 & {\scriptsize interior of $\mathcal{R}_3$}
 & {\scriptsize near boundary}
 & {\scriptsize on $a_1/a_2$ boundary} \\
\midrule
$\VoI(i)$            & $0$       & $11/16$  & $5/2$    \\
$\VoI(j)$            & $3/44$    & $1/16$   & $5/2$    \\
$\VoI(j \mid i)$     & $15/176$  & $9/64$   & $3/32$   \\
\addlinespace
$\E[\Dg]$            & $3/176$   & $49/64$  & ---      \\
$\E[\Dh]$            & $0$       & $11/16$  & ---      \\
\addlinespace
$\Delta\VoI$         & $+3/176$  & $+5/64$  & $-77/32$ \\
\bottomrule
\end{tabular}
\vspace{1em}
\caption{Summary of information values at three marked beliefs. $b_1$: interior complementarity (Theorem~\ref{thm:interior}); $b_2$: boundary complementarity (Proposition~\ref{prop:decomposition}); $b_3$: substitution (Theorem~\ref{thm:converse}). Entries marked~--- indicate that the individual forces are not separately defined at a boundary belief where $V$ has a kink.}
\label{tab:results}
\end{table}

\begin{remark}[PWLC specificity]
\label{rem:pwlc}
The vanishing of $\E[\Dh]$ at $b_1$ relies on $V$ being linear within $\mathcal{R}_3$. This is specific to finite-action problems where decision regions create flat patches in the value function.
\end{remark}

\paragraph{Phase transition.}
Figure~\ref{fig:ray} traces $\Delta\VoI$ along the ray $b(t) = (1/4+t,\; 1/6,\; 7/12-t)$, which walks from $b_2$ toward the $a_1/a_3$ decision boundary. Both forces grow as the boundary approaches, but the substitute force overtakes the complement force at the interaction boundary ($t \approx 0.10$), before the decision boundary itself is reached ($t = 7/60 \approx 0.117$). At $b_1$, action~$a_3$ remains optimal at both of~$i$'s posteriors, so regret is zero and $\E[\Dh] = 0$. At $b_3$, channel~$i$ pushes the belief into $\mathcal{R}_1$ or $\mathcal{R}_2$, where the prior-optimal action incurs large regret, and the substitute force dominates.

\begin{figure}[t]
\centering
\includegraphics[width=0.85\textwidth]{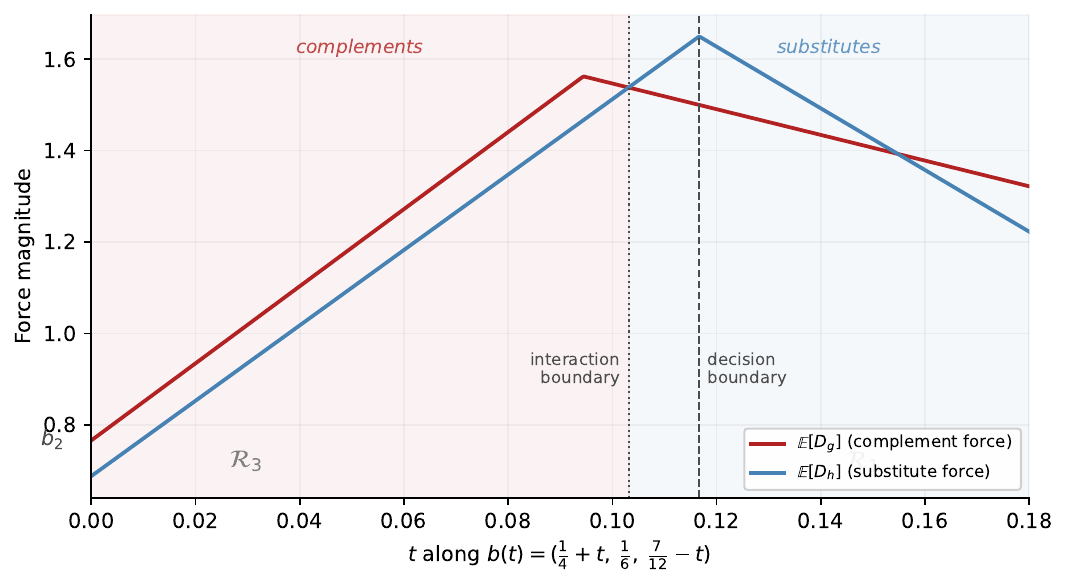}
\caption{Phase transition along $b(t) = (1/4+t,\, 1/6,\, 7/12-t)$. Complement force (red) and substitute force (blue). Shading indicates the complement (red) and substitute (blue) regimes; the interaction boundary (where the forces cross) precedes the decision boundary.}
\label{fig:ray}
\end{figure}

The example illustrates the gap between the necessary and sufficient conditions.

%% file: chapters/related.tex
\section{Related Work}
\label{sec:related}

\paragraph{Information substitutes and complements.}
\citet{ChenWaggoner2017} introduced both global and pointwise notions of information substitutes and complements, including the pointwise interaction term $\Delta\VoI(j \mid i, b)$ that we study.
Their results characterize global conditions, holding across all beliefs, for the two regimes.
Our contribution resolves the local geometry: we identify \emph{which beliefs} support each regime and show that the sign of $\Delta\VoI$ is determined by whether channel~$i$'s posteriors cross a decision boundary.
\citet{Borgers2013} gave a universal characterization via Blackwell comparisons, requiring complementarity or substitutability to hold simultaneously across all decision problems.
Our results are pointwise-in-belief for a fixed decision problem; the global ordering of \citeauthor{Borgers2013} and our local geometry describe complementary facets of the same relationship.
\citet{BrooksFrankelKamenica2024} characterize signal comparisons via a \emph{reveal-or-refine} condition: one signal dominates another in value regardless of the agent's access to other information if and only if every realization of the dominant signal either reveals the state or refines a realization of the dominated signal.

\paragraph{Geometry of the value function and single-channel VoI.}
Several papers exploit the key geometric fact underlying our results: the value function $V$ is linear within each decision region, so VoI vanishes for channels whose posteriors remain in the current region.
\citet{RadnerStiglitz1984} and \citet{ChadeSchlee2002} used this linearity to establish the first-order nonconcavity of VoI with respect to a single channel's precision.
\citet{DeLaraGossner2020} study first-order VoI through convex duality and show that VoI is zero whenever posteriors stay in the interior of the current decision region, the first-order analogue of our Theorem~\ref{thm:interior}.
\citet{Whitmeyer2024} applies the same decision-region linearity to recover a concavity result for single-channel VoI as a function of information quantity.
\citet{KeppoMoscariniSmith2008} extend the \citeauthor{RadnerStiglitz1984} nonconcavity to continuous precision choice, showing that demand for information can be discontinuous and non-monotone.
\citet{AtheyLevin2018} show that in monotone decision problems, where actions are ordered and payoffs supermodular, VoI respects the Blackwell order (a setting where our decision-region geometry reduces to a single threshold).
These papers all study VoI as a function of a single channel's precision; we extend the same geometric engine to the second-order cross-channel interaction $\Delta\VoI(j \mid i, b)$.

\paragraph{Bregman divergences and Bayesian persuasion.}
The martingale identity relating Jensen gaps to expected Bregman divergences was introduced to economic theory by \citet{KamenicaGentzkow2011}, where it writes single-channel VoI as a single expected Bregman divergence.
Applying the same identity to the cross-channel interaction $\Delta\VoI$ instead yields a difference of two such terms, one for each auxiliary value function, which is what gives the decomposition in Proposition~\ref{prop:decomposition} its complement-force-minus-substitute-force structure.
\citet{Chodrow2025} establishes that Bregman divergences are the unique divergences with this information-equivalence property: they are characterized by the agreement between Jensen-gap and divergence-based notions of information content for arbitrary weighted collections, providing an axiomatic basis for the decomposition.
\citet{Mu2021} show that in large samples Blackwell dominance is equivalent to comparisons of R\'{e}nyi divergences, connecting signal ordering to a related family of divergence measures.
\citet{DentiMarinacciRustichini2022} and \citet{Pomatto2023} axiomatize information costs from complementary directions: the former identifies posterior-separable costs (structurally linked to Bregman divergences) as the class consistent with experimental rationality, while the latter characterizes costs with constant marginal returns as those based on R\'{e}nyi entropy.
\citet{Gossner2021} study attention allocation under a posterior-separable capacity constraint, where the geometry of the decision problem determines which states the agent learns to distinguish.

\paragraph{Sequential acquisition and complementarity in dynamic settings.}
\citet{LiangMu2020} study sequential social learning where complementarity between information sources can create learning traps: agents' posteriors evolve into regions where a second source has low marginal value, causing the community to stop acquiring it prematurely.
The mechanism driving their traps is the same boundary-crossing geometry we formalize, where a source is complementary at beliefs whose posteriors do not yet straddle a decision boundary.
Our localization results provide the single-step geometric account of when and why those traps arise.
\citet{GolovinKrause2011} established conditions for global adaptive submodularity of information acquisition policies.
Our results localize that structure: substitution and complementarity depend on the agent's current position in belief space relative to decision boundaries, and global properties emerge from averaging over this local geometry.
\citet{Che2019} study optimal dynamic attention allocation, showing that the optimal policy shifts from diversification to concentration as beliefs sharpen.
\citet{LiangMuSyrgkanis2022} extend the \citeauthor{LiangMu2020} framework to settings with multiple heterogeneous sources, characterizing when delayed aggregation is optimal.

%% file: chapters/conclusion.tex
\section{Conclusion}
\label{sec:discussion}

The localization principle established in this paper says that information sources compete only when one source's posteriors straddle a decision boundary.
Far from any boundary, sources always weakly help each other; substitution is a boundary phenomenon.

\paragraph{Applications.}
The localization principle provides a geometric criterion for information redundancy: two sources substitute only if they resolve the same decision boundary, and sources resolving different boundaries complement each other.
In sequential testing problems (multi-arm clinical trials, multiclass triage), a decision-maker must choose which source to consult next.
Our results yield a concrete heuristic: after observing a source that resolves boundary~$X$, sources targeting a different boundary~$Y$ retain or increase their value (Theorem~\ref{thm:interior}), while sources targeting the same boundary~$X$ lose value.
This localizes the adaptive submodularity conditions of \citet{GolovinKrause2011}, which treat diminishing returns globally, to specific beliefs relative to decision boundaries.
In active learning for multiclass classification, $K$ class labels and a finite classifier family produce exactly the piecewise-linear value structure we study.
Feature redundancy corresponds to resolving the same classification boundary; a boundary-aware selection criterion would complement mutual-information-based approaches, which are agnostic to which boundary a feature resolves.

\paragraph{Limitations.}
Three modeling assumptions bound the scope.
First, the economic interpretation of the complement and substitute forces requires conditional independence between channels, which ensures that $g(\bpost(i, o_i))$ equals the actual expected post-$j$ value; the decomposition and localization theorems themselves hold for arbitrarily correlated channels.
Second, the results are specific to finite-action settings where the value function is piecewise-linear.
Continuous actions yield a smooth, strictly convex~$V$ with no flat decision regions, so the interior complementarity mechanism (Theorem~\ref{thm:interior}) has no direct analogue; the boundary-crossing logic is specific to discrete decisions.
Third, the decomposition requires the Bayesian martingale property $\E[\bpost] = b$; agents with miscalibrated updates break this identity.

\paragraph{Open questions.}
The gap between our necessary and sufficient conditions calls for a tighter characterization: the worked example (Section~\ref{sec:example}) suggests that channels crossing \emph{different} boundaries complement while channels crossing the \emph{same} boundary substitute, but a precise geometric criterion remains to be formulated.
The interaction is generically asymmetric (channel~$i$ may substitute for~$j$ while $j$ complements~$i$), and characterizing when this asymmetry arises would clarify optimal ordering in sequential acquisition.
Quantifying how the complementarity region scales with~$K$ is a natural direction.

%% file: chapters/appendix.tex
\appendix

\section{Piecewise-Linear Structure}
\label{app:polyhedral}

\begin{theoremA}[Piecewise-linearity of the interaction]
\label{prop:polyhedral}
For finite $|A|$, $|O_i|$, $|O_j|$, and $K$, the function $b \mapsto \Delta\VoI(j \mid i, b)$ is piecewise-linear on $\Delta^{K-1}$.
\end{theoremA}

\begin{proof}
For any channel~$k$ with outcome set~$O_k$, the unnormalized posterior-value product satisfies
\[
  P(o \mid b)\, V(\bpost(k, o))
  \;=\; \max_{a \in A}\; \sum_s R(a,s)\, P(o \mid s, k)\, b(s),
\]
where the normalizing denominator $P(o \mid b)$ cancels against the posterior formula.
The right-hand side is the maximum of $|A|$ functions, each linear in~$b$.
Summing over outcomes,
\[
  \E_o\bigl[V(\bpost(k, o))\bigr]
  \;=\;
  \sum_{o \in O_k} \max_{a \in A}\; \sum_s R(a,s)\, P(o \mid s, k)\, b(s),
\]
a sum of piecewise-linear convex (PWLC) functions of~$b$.
By the same cancellation applied to the joint channel $(i,j)$ under conditional independence, and to channels $i$ and $j$ individually, the interaction decomposes as
\[
  \Delta\VoI(j \mid i, b)
  \;=\; F_{ij}(b) - F_i(b) - F_j(b) + V(b),
\]
where each term is a sum of maxima of finitely many linear functions of~$b$.
Each term is therefore PWLC, and their sum is piecewise-linear.
\end{proof}

\section{Regret Interpretation}
\label{app:regret}

In the interior of decision region $\mathcal{R}_a$, where $h(b') = r_a \cdot b'$, the Bregman divergence of $h$ admits a clean interpretation:
\[
  \Dh(b', b)
  \;=\; V(b') - r_a \cdot b'
  \;=\; \operatorname{Regret}(a^*(b),\, b'),
\]
where $\operatorname{Regret}(a, b') = V(b') - r_a \cdot b'$ is the regret of committing to action~$a$ at belief~$b'$.
The substitute force $\E[\Dh(\bpost, b)]$ is therefore the expected regret of sticking with the currently optimal action at the posterior belief.

This connects the Bregman decomposition to a decision-theoretic quantity: substitution occurs when channel~$i$ resolves enough uncertainty that the agent's current action has low expected regret, leaving little room for channel~$j$ to improve the decision.

\section{Lean Formalization}
\label{app:lean}

The Lean~4 source files provide machine-verified proofs of all results.
The formalization uses the \emph{Jensen gap} $\E[f(X)] - f(\E[X])$ rather than explicit Bregman divergences, and abstracts the paper's geometric hypotheses (e.g., posteriors remain in the same decision region) to their algebraic consequences (\texttt{jensenGap h = 0} and \texttt{jensenGap g $\geq$ 0}).
The geometric conditions imply the algebraic ones but not conversely, so the Lean statements are strictly more general.
In particular, the algebraic conditions do not encode conditional independence.

\begin{center}
\footnotesize
\begin{tabular}{@{}ll@{}}
\toprule
Paper & Lean \\
\midrule
Belief $b \in \Delta^{K-1}$ & \texttt{Belief S} \\
Channel likelihood & \texttt{ObsKernel S O} \\
$P(o \mid b)$ & \texttt{marginalProb b k o} \\
$\bpost(i,o)$ & \texttt{posterior b k o hmarg} \\
$\E[D_f(\bpost, b)]$ & \texttt{jensenGap f b k hmarg} \\
\midrule
Prop.~\ref{prop:decomposition} & \texttt{bregman\_decomposition} \\
Cor.~\ref{cor:cost} & \texttt{jensenGap\_add\_const} \\
Martingale property & \texttt{jensenGap\_linear} \\
Thm.~\ref{thm:interior} & \texttt{interior\_complementarity} \\
Thm.~\ref{thm:converse} & \texttt{substitution\_requires\_boundary\_crossing} \\
Jensen's inequality & \texttt{jensenGap\_nonneg\_of\_convexOn} \\
Regret $\geq 0$ & \texttt{regret\_nonneg} \\
Thm.~\ref{prop:polyhedral} cancellation & \texttt{unnorm\_posterior\_value} \\
Thm.~\ref{prop:polyhedral} PWLC expansion & \texttt{expected\_value\_as\_sup\_linear} \\
\bottomrule
\end{tabular}
\end{center}